\documentclass[11pt]{article}
\usepackage[a4paper,margin=27mm]{geometry}
\usepackage{amsmath,amssymb,amsthm,mathtools}
\usepackage{booktabs,microtype}
\usepackage{algorithm,algpseudocode}
\usepackage[T1]{fontenc}
\usepackage{lmodern}
\usepackage[hidelinks]{hyperref}
\hypersetup{pdftitle={Directed Hamiltonian-Cycle Parity in O*((3/2)\string^n) Deterministic Time and Polynomial Space},
  pdfauthor={Hanqing Li}}
\numberwithin{equation}{section}
\newtheorem{theorem}{Theorem}[section]
\newtheorem{lemma}[theorem]{Lemma}
\newtheorem{proposition}[theorem]{Proposition}
\newtheorem{corollary}[theorem]{Corollary}
\theoremstyle{definition}

\newcommand{\F}{\mathbb F_2}
\newcommand{\Zthree}{\mathbb Z_3}
\newcommand{\Os}{O^*}
\newcommand{\one}{\mathbf 1}
\newcommand{\E}{\mathbb E}
\newcommand{\Var}{\operatorname{Var}}
\newcommand{\supp}{\operatorname{supp}}
\newcommand{\diag}{\operatorname{diag}}
\newcommand{\rank}{\operatorname{rank}}
\newcommand{\im}{\operatorname{im}}
\newcommand{\Ham}{\operatorname{Ham}}
\newcommand{\Own}{\operatorname{own}}
\newcommand{\Even}{\mathcal E}
\newcommand{\calQ}{\mathcal Q}

\allowdisplaybreaks[1]

\title{Directed Hamiltonian-Cycle Parity in\\
  $\boldsymbol{O^*((3/2)^n)}$ Deterministic Time and Polynomial Space}
\author{Hanqing Li\\Peking University}
\date{}

\begin{document}
\maketitle

\begin{abstract}
We give a deterministic algorithm that computes the parity of the number
of Hamiltonian cycles in an $n$-vertex directed graph in
$O(n^4(3/2)^n)$ time and $O(n^2)$ bits of working space, improving the
$O^*(\varphi^n)$ bound of Bj\"orklund and Husfeldt. Their local-degree
formula reduces the problem to a weighted sum over solutions of structured
quadratic equations. We cover the corresponding ternary state space by
binary subcubes, each inducing an affine system. The Kuang--Wang cover can
be regenerated within the target bound; canonical ownership resolves its
overlaps, while self-loop conditional expectations bound every affine
solution visit. Rollback elimination shares the work across cover prefixes.
The same cover gives a Las Vegas algorithm listing all $L$ solutions of
$m$ affine product constraints in $N$ Boolean variables in expected time
$\operatorname{poly}(N,m)((3/2)^m+L)$ and polynomial space.
Finally, we show that complete enumeration can require
$\Omega((3/2)^n)$ visits even on strongly connected digraphs after an
optimal self-loop choice. This is a limitation of the enumeration method,
not a general lower bound for Hamiltonian-cycle parity.
\end{abstract}

\section{Introduction}
\label{sec:introduction}

Counting directed cycle covers modulo two is a determinant computation.
Requiring a cover to consist of one cycle adds a connectivity constraint.
Bj\"orklund and Husfeldt~\cite{BH13} expressed Hamiltonian-cycle parity
using only local outdegrees and obtained a deterministic
$\Os(\varphi^n)$-time, polynomial-space algorithm for general digraphs,
where $\varphi=(1+\sqrt5)/2$. On bipartite digraphs they obtained expected
time $\Os((3/2)^n)$. We attain the latter exponential base deterministically
on unrestricted digraphs, while retaining polynomial space.

\begin{theorem}\label{thm:main}
The parity of the number of Hamiltonian cycles in an $n$-vertex directed
graph can be computed deterministically in $O(n^4(3/2)^n)$ time and
$O(n^2)$ bits of working space.
\end{theorem}

The bound includes construction and every regeneration of the covering
family. We use ordinary binary elimination; the polynomial factor does
not rely on fast matrix multiplication.

\paragraph{The algorithmic interface.}
The local-degree formula of~\cite{BH13} has two computational parts,
called P2 and P3. P2 lists the vectors satisfying $x\circ(Ax)=x$.
P3 assigns each nonzero vector a bit by testing consistency and uniqueness
of an affine system. Summing these bits gives the desired parity.
We retain P3 and improve the enumeration of P2.

Delete the input diagonal to obtain $B$, and write
$A=B+\diag(c+\one)$. Setting $y=Bx+c$ transforms P2 into
\[
 Bx+y=c,\qquad (x_i,y_i)\in\{(0,0),(1,0),(0,1)\}\quad(i\in[n]).
\]
Each pair has three legal states. Excluding any one leaves an affine line
over $\F$, so a cover of the ternary cube by binary subcubes gives a family
of affine systems. For any such cover with $M$ centers, the total number
$V(c)$ of affine solution visits, including repetitions, satisfies
$\E_cV(c)=M$. Conditional expectation selects a diagonal with $V(c)\le M$.
Changing the diagonal preserves Hamiltonian cycles for $n\ge2$.

The cover must support more than a small cardinality bound. We need to
generate it repeatedly without an exponential table, retain each state
at exactly one covering center, and pay for discarded visits as well as
retained ones. Kuang and Wang~\cite{KW26} construct a cover from the even
induced vertex sets of an auxiliary graph. Geometrically decreasing blocks
make its construction and repeated traversal fit the target bound in
polynomial space. A deterministic choice in their consistency proof gives
the required ownership map. Sharing column elimination along this product
cover yields the fourth-degree polynomial factor.

The local-degree formula, P3, and the self-loop conditional-expectation
principle are due to Bj\"orklund and Husfeldt~\cite[Sections 2 and 4.1]{BH13}.
The graph-based cover and its consistency argument are due to Kuang and
Wang~\cite{KW26}. Our contribution is to make the cover usable within this
parity framework with complete bounds on construction, regeneration,
overlap handling, and all solution visits.

\paragraph{Further results and related work.}
The same mechanism applies to affine product constraints
$u_i(\xi)v_i(\xi)=0$. With $m$ constraints in $N$ variables and $L$
solutions, we obtain polynomial-space Las Vegas enumeration in expected
time $\operatorname{poly}(N,m)((3/2)^m+L)$
(Theorem~\ref{thm:general}). Feasibility and search take deterministic
$\operatorname{poly}(N,m)(3/2)^m$ time. A single random ternary translation
controls repeated visits in the enumeration algorithm; every translation
gives a correct, duplicate-free output.

These constraints are the two-factor case of PAF-SAT, equivalently
2-Sub-Sat, studied by Arvind and Guruswami~\cite{AG21}. Their
$\Os((3/2)^r)$ randomized search algorithm, with exponent measured in
subspace dimension $r$, uses the same local choices
$u=0$, $v=0$, and $u+v=1$; see Theorem 15 and its proof in the full version.
Our bound is measured in the number of constraints and concerns complete
enumeration. It does not derandomize their dimension bound for an arbitrary
number of constraints. Finding one P2 solution would not suffice here:
$x=0$ is always feasible, while the Hamiltonian answer is a weighted sum
over nonzero solutions.

For general quadratic systems, Dinur~\cite[Theorem 1.2]{Dinur21} lists all
solutions with high probability in
$O(\max\{2^{0.6943N},L2^{\varepsilon N}\})$ time for fixed $\varepsilon>0$.
Our input class is more restricted, but the enumeration has zero error,
polynomial workspace, and polynomial overhead per output in total time.
Its delay between outputs may be exponential.

Finally, we prove an output-size barrier for complete P2 enumeration,
uniformly over all choices of artificial self-loops
(Theorem~\ref{thm:lower}). This identifies the exponential base of that
method, while leaving open faster ways to aggregate the weighted sum.

\paragraph{Conventions.}
Vertices are ordered as $[n]$ and adjacency rows index arc tails.
A Hamiltonian cycle is counted once by fixing its starting vertex;
reversed orientations are distinct. The graph has at most one arc per
ordered pair, and may have self-loops. For $n=1$ we return the input loop
bit. All linear algebra is over $\F$. Counts, expectations, and ranks
used as exponents are integers or rationals. $O^*$ suppresses polynomial
factors. For explicit bounds we count binary operations on a RAM with
$O(\log n)$-bit index words and charge long-integer arithmetic by bit
length. Space means working space; listed solutions go to an output stream.

\section{The local-degree interface}
\label{sec:local}

For a vertex set $S$, let $d_i(S)=\sum_{j\in S}a_{ij}$, reduced modulo
two when used in an expression over $\F$. For disjoint $X,Y$ that are not
both empty, write $X\prec Y$ if $\min X<\min Y$, with
$\min\varnothing=+\infty$. The parity specialization of the local-degree
formula of Bj\"orklund and Husfeldt~\cite{BH13} is
\begin{equation}\label{eq:local2}
 H(A):=|\Ham(A)|\bmod2
 =\sum_{\substack{X,Y,Z\text{ partition }V\\X\prec Y}}
 \prod_{x\in X}d_x(X)\prod_{y\in Y}d_y(Y)
 \prod_{z\in Z}d_z(X\cup Y)\quad\text{in }\F.
\end{equation}
Here $X$ is nonempty and empty products equal one. This identity results
from pairing the two color classes in the integer formula before dividing
by two; no division in $\F$ is involved.

For a fixed nonempty $X$, define
\begin{equation}\label{eq:fdef}
 f_A(X)=\sum_{\substack{Y\subseteq V\setminus X\\X\prec Y}}
       \prod_{y\in Y}d_y(Y)
       \prod_{z\in V\setminus(X\cup Y)}d_z(X\cup Y)\in\F.
\end{equation}
The factor depending only on $X$ in~\eqref{eq:local2} is one precisely
when its indicator vector $x$ satisfies $x\circ(Ax)=x$.
This is the P2 feasibility condition. The next lemma is the P3 computation.

\begin{lemma}[Computing the contribution]\label{lem:p3}
For nonempty $X$, $f_A(X)$ equals one if and only if the following affine
system in $\eta\in\F^n$ has exactly one solution:
\begin{align}
 \eta_i&=0 &&\text{if }i\in X\text{ or }i\le\min X,
       \label{eq:p3fixed}\\
 \sum_{j=1}^n a_{ij}\eta_j+d_i(X)\eta_i&=1+d_i(X)
       &&\text{if }i\notin X.\label{eq:p3rows}
\end{align}
Consequently $f_A(X)$ is computable in polynomial time and space.
\end{lemma}

\begin{proof}
Interpret $\eta$ as the indicator of $Y$. Equations~\eqref{eq:p3fixed}
give disjointness and the order $X\prec Y$, also allowing $Y$ to be
empty. For $i\in Y$, equation~\eqref{eq:p3rows} reduces to $d_i(Y)=1$.
For $i\notin X\cup Y$, it reduces to $d_i(X)+d_i(Y)=1$.
These are exactly the conditions that all factors of the corresponding
summand of~\eqref{eq:fdef} are odd. Thus $f_A(X)$ is the parity of the
number of solutions of this affine system. An affine system over $\F$
has either no solutions or $2^d$ solutions for some $d\ge0$, so its
solution count is odd exactly when it is a singleton. Gaussian elimination
tests consistency and uniqueness.
\end{proof}

Define
\begin{equation}\label{eq:PB}
 \mathcal P_B(c)=\{x\in\F^n:x_i((Bx)_i+c_i)=0\ (i\in[n])\},
 \qquad K_B(c)=|\mathcal P_B(c)|.
\end{equation}

Delete the input diagonal to obtain $B$, and choose any $c\in\F^n$.
Set $A=B+\diag(c+\one)$. The Boolean identity $x_i^2=x_i$ gives
\begin{equation}\label{eq:diagonalP2}
 x\circ(Ax)=x
 \quad\Longleftrightarrow\quad
 x_i((Bx)_i+c_i)=0\quad\text{for every }i.
\end{equation}
Since a Hamiltonian cycle on at least two vertices uses no self-loop,
\begin{equation}\label{eq:weighted}
 H(B)=H(A)
 =\sum_{\substack{x\in\mathcal P_B(c)\\x\ne0}}f_A(\supp x)
 \quad(n\ge2).
\end{equation}
The diagonal affects both the feasible sets and their weights, but the
weighted sum is invariant. The parity of $K_B(c)$ alone is not the target.

\section{From a cover to a weighted affine fiber}
\label{sec:parity}

Identify the
three symbols $0,1,2$ with the three points
\begin{equation}\label{eq:sigma}
 \Sigma=\{(0,0),(1,0),(0,1)\}\subseteq\F^2
\end{equation}
in this order. For $q\in\Zthree^m$ define
\[
 Q_q=\{s\in\Zthree^m:s_i\ne q_i\text{ for all }i\}.
\]
We call $q$ a center and $Q_q$ its binary subcube. A family of centers
$\calQ$ is a cover if $\bigcup_{q\in\calQ}Q_q=\Zthree^m$.
Under~\eqref{eq:sigma}, each $Q_q$ is an $m$-dimensional affine subspace
of $\F^{2m}$: any two distinct points over $\F$ form an affine line.

Fix the zero-diagonal input matrix $B$. For $x\in\F^n$, introduce
$y=Bx+c$. Equations~\eqref{eq:PB} and~\eqref{eq:sigma} give the bijection
\begin{equation}\label{eq:encoding}
 x\in\mathcal P_B(c)
 \quad\Longleftrightarrow\quad
 Bx+y=c,\quad (x_i,y_i)\in\Sigma\text{ for all }i.
\end{equation}
It is a bijection because $x$ determines $y$ once $B,c$ are fixed.
Let $\pi_B(s)=Bx+y$ for the pair $(x,y)$ represented by $s$.
The legal states for P2 are the fiber $\pi_B^{-1}(c)$.

\subsection{Restricting a fiber to a binary subcube}

For a center $q$, use a Boolean parameter $z_i$ for the two permitted
states at coordinate $i$. Table~\ref{tab:affine} fixes the parametrization.
Write $B_i=B[:,i]$ and let $e_i$ be the $i$th unit vector.

\begin{table}[ht]
\centering
\caption{Affine parametrization after prohibiting state $q_i$.}
\label{tab:affine}
\begin{tabular}{cccc}
\toprule
$q_i$ & $(x_i,y_i)$ & column $i$ of $M_q$ & contribution to $d_q$\\
\midrule
$0$ & $(1+z_i,z_i)$ & $B_i+e_i$ & $B_i$\\
$1$ & $(0,z_i)$ & $e_i$ & $0$\\
$2$ & $(z_i,0)$ & $B_i$ & $0$\\
\bottomrule
\end{tabular}
\end{table}

In particular,
\begin{equation}\label{eq:affinesystem}
 \pi_B(s)=M_qz+d_q,\qquad
 d_q=\sum_{i:q_i=0}B_i.
\end{equation}
The maps $z\mapsto s$ and $s\mapsto z$ are inverse on $Q_q$, regardless
of whether $M_q$ is singular. Hence the solutions of
\begin{equation}\label{eq:system}
 M_qz=c+d_q
\end{equation}
are in bijection with the states in $Q_q\cap\pi_B^{-1}(c)$.
Each system has $n$ variables and $n$ equations. Both $M_q$ and $d_q$
are independent of the right-hand side $c$.

\subsection{Controlling all solution visits by choosing the diagonal}

Let $\calQ$ be any fixed binary-subcube cover of $\Zthree^n$, with
$M=|\calQ|$. At this stage no graph-based construction is assumed. Define the integer
\begin{equation}\label{eq:visits}
 V(c)=\sum_{q\in\calQ}\#\{z\in\F^n:M_qz=c+d_q\}.
\end{equation}
This counts every visit in every affine system, including repetitions of
the same state in overlapping subcubes.

\begin{lemma}[A diagonal with bounded total output]\label{lem:diagonal}
Suppose a representation of $\calQ$ can enumerate its centers without
repetition in time $R$ and working space $S$. Then a vector $c$ satisfying
$V(c)\le M$ can be found deterministically in $O(nR+n^4M)$ time and
$S+O(n^2)$ bits of working space, provided $M\le3^n$.
The bound is in elementary binary operations; Section~\ref{sec:implementation}
improves it by sharing elimination across a product cover.
\end{lemma}

\begin{proof}
For a uniform $c$, each fixed pair $(q,z)$ satisfies~\eqref{eq:system}
with probability $2^{-n}$. Thus
\begin{equation}\label{eq:visitmean}
 \E_c V(c)=\sum_{q\in\calQ}\sum_{z\in\F^n}2^{-n}=M.
\end{equation}
Suppose the prefix $c_{\le\ell}=(c_1,\ldots,c_\ell)$ is fixed. Let
$Z_q(c_{\le\ell})$ count the assignments $z\in\F^n$ satisfying the
first $\ell$ rows of~\eqref{eq:system}. For any such assignment the
remaining $n-\ell$ right-hand-side bits have one required value, so
\begin{equation}\label{eq:prefixCE}
 \E[V(c)\mid c_{\le\ell}]
 =2^{-(n-\ell)}\sum_{q\in\calQ}Z_q(c_{\le\ell}).
\end{equation}
For a given $q$, Gaussian elimination computes this count as zero for an
inconsistent prefix system, and otherwise as
$2^{n-\rank M_q[1:\ell,:]}$.

To fix the next bit, regenerate all centers once, computing and summing
the counts for both candidate prefixes. Their denominators in
\eqref{eq:prefixCE} agree, so exact integer comparison suffices.
Choose the candidate with the smaller numerator, breaking ties by zero.
The conditional expectation never increases, and at the final prefix
it equals $V(c)$, proving $V(c)\le M$.

There are $n$ complete traversals, and ordinary elimination takes
$O(n^3)$ time per center and candidate. Each numerator is at most $M2^n$,
so its bit length is $O(n)$ under the stated size bound; additions and
comparisons cost $O(n)$ time. Only the current center and its elimination
workspace are retained, proving the claimed bounds.
\end{proof}

The distinction between $K_B(c)$ and $V(c)$ matters. A small number of
distinct feasible states would not itself control the work caused by
overlapping subcubes. Lemma~\ref{lem:diagonal} bounds the quantity that
the final traversal actually pays for.

After choosing $c$, enumerate each affine system and retain a state only
at a canonical covering center $\Own(s)$. If this map is computable in
polynomial time and space, every P2 solution is processed exactly once
without storing a visited-state list. The next section supplies such a
map together with an efficiently regenerable cover.

\section{Deterministic covers that can be regenerated in polynomial space}
\label{sec:covers}

We construct covers independently of the input digraph.

\subsection{The Kuang--Wang construction and an ownership map}

Let $L$ be the zero-diagonal adjacency matrix of a simple undirected graph
on $[b]$, and define its even induced vertex sets by
\[
 \Even(L)=\{S\subseteq[b]:(L\one_S)_i=0\text{ for every }i\in S\}.
\]
Following Kuang and Wang~\cite[Lemmas 2.1 and 2.2]{KW26}, associate to
$S\in\Even(L)$ the center
\begin{equation}\label{eq:kwcenter}
 a_i^S=\begin{cases}
 0,&i\in S,\\
 1,&i\notin S\text{ and }(L\one_S)_i=1,\\
 2,&i\notin S\text{ and }(L\one_S)_i=0.
 \end{cases}
\end{equation}

\begin{lemma}[Cover and ownership]\label{lem:kw}
The centers $a^S$, $S\in\Even(L)$, are distinct and cover $\Zthree^b$.
There is a deterministic polynomial-time map $\Own_L$ assigning each
$s\in\Zthree^b$ to a center with $s\in Q_{\Own_L(s)}$.
\end{lemma}

\begin{proof}
We use the following elementary fact: if $M$ is a symmetric binary matrix,
then its diagonal vector $\delta(M)$ belongs to $\im M$.
Indeed, for $w\in\ker M$, cancellation of symmetric off-diagonal terms gives
\[
 w^T\delta(M)=w^TMw=0.
\]
Since $\im M=(\ker M)^{\perp}$, the assertion follows.
This is the linear-algebra form of the handshake argument in~\cite{KW26}.

Fix $s$, put $V_j=\{i:s_i=j\}$ and $U=V_1\cup V_2$, and solve
\begin{equation}\label{eq:owner}
 \bigl(L[U]+\diag(\one_{V_2})\bigr)u=\one_{V_2}
 \quad\text{on }U.
\end{equation}
Here $L[U]$ is the principal submatrix on $U$. The right side is the
diagonal vector of the symmetric matrix on the left, so the system is
consistent, including when $U$ is empty. Extend $u$ by zero outside $U$
and let $S=\supp u$.

For $i\in V_1$, the equation gives $(Lu)_i=0$; for $i\in V_2$ it
gives $(Lu)_i=1+u_i$. Whenever $u_i=1$, either equation yields $(Lu)_i=0$,
and hence $S\in\Even(L)$. If $i\in V_0$, then $i\notin S$ and
$a_i^S\ne0$. If $i\in V_1$, then $a_i^S\in\{0,2\}$; if $i\in V_2$,
then $a_i^S\in\{0,1\}$. Thus $a_i^S\ne s_i$ in every coordinate.
Finally, $S$ is exactly the zero set of $a^S$, proving distinctness.

Fix the vertex order and a deterministic row-reduction convention, and
set all free variables of~\eqref{eq:owner} to zero. Define $\Own_L(s)$
to be the center obtained from that particular solution. This specifies
one center without searching the center family or storing an ownership table.
\end{proof}

\subsection{Choosing each auxiliary graph without a subset table}

\begin{lemma}\label{lem:localcover}
For every $b\ge1$, one can deterministically construct an undirected
graph $L_b$ such that
\begin{equation}\label{eq:localsize}
 |\Even(L_b)|\le2(3/2)^b-1
\end{equation}
in time $\operatorname{poly}(b)2^b$ and polynomial space. After construction,
all its centers can be enumerated without repetition in
$\operatorname{poly}(b)2^b$ total time and polynomial space.
\end{lemma}

\begin{proof}
For a uniformly random graph $L\sim G(b,1/2)$ and nonempty $S$, the degree
parities in $L[S]$ are the image of the random edge vector under the binary
incidence matrix of the complete graph on $S$. This matrix has rank
$|S|-1$: its image is the subspace of vectors whose coordinate sum is zero.
The assertion also holds for $|S|=1$, with rank zero. Consequently,
\begin{equation}\label{eq:kwmean}
 \E|\Even(L)|
 =1+\sum_{s=1}^b\binom bs2^{-(s-1)}
 =2(3/2)^b-1.
\end{equation}

Expose the $\binom b2$ edge bits in a fixed order. We show how to compute
the exact conditional expectation after any partial assignment. Fix $S$,
and let $F_S$ be the graph on $S$ whose edges are the still unexposed pairs.
Let $h\in\F^S$ record the degree parities contributed by exposed edges
whose assigned value is one. The remaining parity constraints are the
incidence system $\partial_{F_S}e=h$. Its image consists of the vectors
whose sum on each connected component is zero, and its rank is
$|S|-\kappa(F_S)$, where isolated vertices count as components. Therefore
\begin{equation}\label{eq:localCE}
 \Pr[S\in\Even(L)\mid\text{exposed edges}]
 =\begin{cases}
 2^{-(|S|-\kappa(F_S))},
   &\sum_{i\in C}h_i=0\text{ for every component }C,\\
 0,&\text{otherwise}.
 \end{cases}
\end{equation}
The empty set has probability one. The incidence-matrix assertion follows,
for example, by solving on a spanning tree of each component; all but
one vertex parity can be prescribed independently.

For each of the two candidate values of the next edge, scan all $2^b$
subsets $S$, sum~\eqref{eq:localCE}, and choose the value giving the smaller
sum, breaking ties by zero. The current conditional expectation is the
average of the two candidate expectations, so it never increases.
At termination it is $|\Even(L_b)|$, proving~\eqref{eq:localsize}.

All probabilities have denominator dividing $2^b$. After multiplying by
$2^b$, their sum is an integer at most $2^{2b}$, with $O(b)$ bits.
Each subset is processed and discarded using polynomial space.
There are polynomially many edge decisions and polynomial work per subset,
giving the stated time with exponential factor exactly $2^b$.
Retain only the resulting graph. To enumerate its centers, scan all
subsets again, test evenness, and output~\eqref{eq:kwcenter} on success.
\end{proof}

\subsection{Geometric blocks and the cost of regeneration}

The local bound $2(3/2)^b$ carries a constant factor per block, and the
local enumerator scans $2^b$ subsets even when it outputs far fewer centers.
The following block schedule controls both costs simultaneously.

\begin{theorem}[A regenerable cover]\label{thm:streamcover}
For every $m\ge1$, a cover $\calQ_m$ of $\Zthree^m$ can be represented
using polynomial space such that
\begin{equation}\label{eq:globalsize}
 M:=|\calQ_m|\le2m(3/2)^m.
\end{equation}
Its representation can be constructed, and its centers enumerated once
each, in $\operatorname{poly}(m)(3/2)^m$ time and polynomial space.
The same bound holds for each subsequent complete enumeration from the
representation. A canonical covering center $\Own(s)$ is computable in
polynomial time for every state $s$.
\end{theorem}

\begin{proof}
Write $a=3/2$. Starting with $t_1=m$, define
\[
 b_j=\lceil t_j/2\rceil,\qquad t_{j+1}=\lfloor t_j/2\rfloor
\]
until the remaining length is zero. This gives
$g=\lfloor\log_2m\rfloor+1$ blocks and $2^g\le2m$.
For example, $m=16$ gives block lengths $8,4,2,1,1$.
Construct the graph of Lemma~\ref{lem:localcover} on each block and
write $m_j=|\Even(L_{b_j})|\le2a^{b_j}$.
Use the Cartesian product of the block center families as $\calQ_m$.
The cover property follows blockwise, and
\[
 M=\prod_{j=1}^gm_j\le2^ga^m\le2ma^m.
\]
The blockwise ownership map from Lemma~\ref{lem:kw} gives $\Own$.

Enumerate the product by nested loops with the largest block outermost.
At level $j$, scan all $2^{b_j}$ subsets of that block and descend to
level $j+1$ only when the current subset is even. Each inner loop is
restarted from the beginning, with no cache. The exact number of subset
tests in one complete traversal is
\begin{equation}\label{eq:scanW}
 W=\sum_{j=1}^g\left(\prod_{i<j}m_i\right)2^{b_j}.
\end{equation}
For every integer $t\ge1$,
\begin{equation}\label{eq:halving}
 2^{\lceil t/2\rceil}\le\frac43a^t.
\end{equation}
For $t=2k$ the ratio of the left side to $a^t$ is $(8/9)^k$;
for $t=2k+1$ it is $(4/3)(8/9)^k$. Since the preceding blocks have
total length $m-t_j$, equations~\eqref{eq:scanW} and~\eqref{eq:halving} give
\begin{equation}\label{eq:scanbound}
 W\le\frac43a^m\sum_{j=1}^g2^{j-1}
   =\frac43(2^g-1)a^m<\frac83ma^m.
\end{equation}
This accounts for failed evenness tests and every restart. Each test
and each output requires only polynomial additional work.

Constructing all the graphs costs at most
$\operatorname{poly}(m)\sum_j2^{b_j}
 \le\operatorname{poly}(m)2^{\lceil m/2\rceil}$,
which satisfies the desired bound by~\eqref{eq:halving}.
The graphs occupy $O(\sum_jb_j^2)=O(m^2)$ bits.
The loop stack, current subsets, current center, and ownership computations
require polynomial space. Restarting the traversal changes none of these
bounds.
\end{proof}

The theorem bounds complete traversals, not the delay between consecutive
centers. That distinction is sufficient: the algorithm below makes only
polynomially many complete traversals. Geometric block sizes and their
nesting order are part of the construction, rather than an implementation
choice that can be omitted from the time analysis.

\section{The deterministic algorithm and rollback implementation}
\label{sec:implementation}

Instantiate Section~\ref{sec:parity} with the cover $\calQ_n$ of
Theorem~\ref{thm:streamcover}. The algorithm below specifies the
computation independently of the elimination data structure; rollback
implements its diagonal-selection step more efficiently.

\subsection{Enumeration, ownership, and contribution accumulation}

An affine solution space can be streamed by storing one particular
solution and a basis of its nullspace. Enumerating all binary combinations
of the basis vectors visits each affine solution once, with polynomial
work per visit and polynomial memory. For a generated state $s$, compute
$\Own(s)$ blockwise, and process $s$ only when the current center is
its owner. This does not impose an additional linear constraint on
\eqref{eq:system}; it is a polynomial-time test after a solution has
been generated.

\begin{algorithm}[ht]
\caption{Directed Hamiltonian-cycle parity}
\label{alg:main}
\begin{algorithmic}[1]
\Require Adjacency matrix of an $n$-vertex digraph, $n\ge2$
\State Delete the input diagonal to obtain $B$.
\State Construct the block graphs representing $\calQ_n$
       (Theorem~\ref{thm:streamcover}).
\State Choose $c$ by the prefix conditional expectations of
       Lemma~\ref{lem:diagonal}.
\State $A\gets B+\diag(c+\one)$; $h\gets0$.
\For{each center $q$ regenerated from $\calQ_n$}
  \State Form $M_q,d_q$ using Table~\ref{tab:affine}.
  \For{each affine solution $z$ of $M_qz=c+d_q$, streamed by elimination}
    \State Decode $z$ to its ternary state $s$ and compute $\Own(s)$.
    \If{$q=\Own(s)$}
      \State $X\gets\{i:s_i=1\}$.
      \If{$X\ne\varnothing$}
        \State $h\gets h+f_A(X)$ in $\F$, using Lemma~\ref{lem:p3}.
      \EndIf
    \EndIf
  \EndFor
\EndFor
\State \Return $h$
\end{algorithmic}
\end{algorithm}

\begin{proof}[Proof of Theorem~\ref{thm:main}]
Consider Algorithm~\ref{alg:main}. Every affine solution it generates
decodes to a legal state of syndrome $c$ by~\eqref{eq:affinesystem}.
Conversely, every state in $\pi_B^{-1}(c)$ lies in at least one covering
subcube and therefore is generated. Its owner is one of these covering
centers, so it is retained exactly once. The bijection~\eqref{eq:encoding}
then shows that the retained states correspond exactly to
$\mathcal P_B(c)$. Lemma~\ref{lem:p3} computes their nonempty-set
contributions, and~\eqref{eq:weighted} proves that the returned bit is
$H(B)$, equal to the answer for the input.

Proposition~\ref{prop:resources} proves the stated time and space bounds.
The one-vertex case is handled by its input loop bit.
\end{proof}

\subsection{Sharing column elimination along the product traversal}

Let $a=3/2$, $P_j=\prod_{i\le j}m_i$, $P_0=1$, and $M=P_g$.
Write $t_j=\sum_{i\ge j}b_i$ for the residual length, with $t_{g+1}=0$.
The cover-volume bound and the local construction imply
\begin{equation}\label{eq:localbounds}
 a^{b_j}\le m_j\le2a^{b_j}.
\end{equation}
At a traversal node, the columns of $M_q$ for previously chosen blocks
remain unchanged when the next block center is selected.

\begin{lemma}[Amortized column insertions]\label{lem:insertions}
Insert a block's columns when entering its node and remove them on return.
For the halving cover, the number of column insertions in a full traversal is
\[
 I=\sum_{j=1}^g b_jP_j\le9M.
\]
More generally, $I\le9CM$ whenever $b_j\le C(t_{j+1}+1)$ for a constant $C$.
\end{lemma}
\begin{proof}
There are $P_j$ nodes at level $j$. Equation~\eqref{eq:localbounds} gives
$M/P_j\ge a^{t_{j+1}}$. The residuals $t_{j+1}$ are distinct nonnegative
integers, and hence
\[
 \frac IM\le\sum_j b_ja^{-t_{j+1}}
 \le C\sum_{t\ge0}(t+1)(2/3)^t=9C.
\]
For halving, $b_j\le t_{j+1}+1$.
\end{proof}

\begin{lemma}[One diagonal pass]\label{lem:pass}
Set $W_2=\sum_jP_{j-1}b_j^2 2^{b_j}$. Both candidate numerators at one
diagonal position can be computed in
$O(W_2+n^2(I+M))$ time and $O(n^2)$ bits of workspace.
\end{lemma}
\begin{proof}
At position $\ell$, maintain an unreduced echelon basis of the columns
already selected, restricted to rows $1,\ldots,\ell$. Reduce an inserted
vector against the stored pivot vectors. If its residue is nonzero, put it
in the previously empty slot of its pivot. Existing basis vectors never
change. An insertion takes $O(n^2)$ binary operations. Record the new
pivot slot on a stack; returning to a saved mark removes just those slots.
There are at most $n$ stored vectors at any instant. No matrix snapshot
is saved for each node.

Maintain $d_q$ by adding or removing $B[:,i]$ whenever a coordinate with
$q_i=0$ is entered or left. At a leaf with rank $r$, reduce the two
candidate right-hand sides against the basis. A zero residue contributes
$2^{n-r}$; a nonzero residue contributes zero. The two membership tests
take $O(n^2)$ time, and the $O(n)$-bit counters fit within this cost.
Evenness tests for the scanned subsets cost $O(W_2)$ in total.
The input, graphs, basis, shift, and rollback stack require $O(n^2)$ bits.
\end{proof}

\begin{proposition}[Explicit resource bounds]\label{prop:resources}
Algorithm~\ref{alg:main} can be implemented in
$O(n^4(3/2)^n)$ deterministic time and $O(n^2)$ bits of working space
on a RAM with $O(\log n)$-bit index words. Arithmetic on longer integers
is charged according to its bit length.
\end{proposition}
\begin{proof}
Use ordinary binary elimination. A system of order $n$ costs $O(n^3)$
time and $O(n^2)$ bits; fixed variables in P3 are first substituted out.
For a block of size $b$, the component tests in~\eqref{eq:localCE} take
$O(b^2)$ time per subset, at $O(b^2)$ edge decisions. Thus all graph
construction costs
\[
 O\!\left(\sum_j b_j^4 2^{b_j}\right)
 \le O\!\left(n^4 2^{\lceil n/2\rceil}\right)
 =O(n^4a^n),
\]
where $\sum_jb_j^4\le n^4$ and~\eqref{eq:halving} account for construction
without an extra factor for the number of blocks.

Theorem~\ref{thm:streamcover} gives $M,W=O(na^n)$, where $W$ is
the subset count in~\eqref{eq:scanW}. Therefore $W_2\le n^2W=O(n^3a^n)$.
Lemmas~\ref{lem:insertions} and~\ref{lem:pass} bound a diagonal pass by
$O(n^3a^n)$ and all $n$ passes by $O(n^4a^n)$.

The final traversal has $M$ affine systems and at most $V(c)\le M$
solution visits. Ordinary affine solving, ownership, and P3 each cost
at most $O(n^3)$ per relevant center or visit. Streaming from a stored
nullspace basis costs at most $O(n^2)$ per solution. The final cost is
therefore $O(n^3(M+V(c))+W_2)=O(n^4a^n)$.

The input, block graphs, one rollback basis, and a constant number of
linear-algebra workspaces occupy $O(n^2)$ bits. Subsets and centers use
$O(n)$ bits; pivot indices and marks use $O(n\log n)$ bits; exact counters
use $O(n)$ further bits. This proves the bound without local center tables.
\end{proof}

The gain comes from sharing columns during diagonal selection. It does
not assume a quadratic P3 update under a Gray-code step: a nullspace basis
vector can be dense, changing many entries of the weight system.

\section{Solving and listing affine product constraints}
\label{sec:listing}

\begin{theorem}\label{thm:general}
Let $u_i,v_i:\F^N\to\F$ be affine forms, $i\in[m]$, and let $L$ be the
number of solutions of $u_i(\xi)v_i(\xi)=0$ for all $i$. All solutions
can be listed without repetition in polynomial working space and expected
time $\operatorname{poly}(N,m)((3/2)^m+L)$ by a Las Vegas algorithm.
\end{theorem}

At each coordinate
the legal pair $(u_i(\xi),v_i(\xi))$ belongs to $\Sigma$.
Excluding symbols $0,1,2$, respectively, is equivalent to imposing
\begin{equation}\label{eq:generalrows}
 u_i(\xi)+v_i(\xi)=1,\qquad
 u_i(\xi)=0,\qquad
 v_i(\xi)=0.
\end{equation}
Each equation enforces legality as well as exclusion.
The first permits $(1,0)$ and $(0,1)$; the second permits $(0,0)$ and
$(0,1)$; the third permits $(0,0)$ and $(1,0)$.
Thus each center gives an affine system in the original $N$ variables.

\begin{corollary}[Deterministic feasibility and search]\label{cor:decision}
Given the affine forms of Theorem~\ref{thm:general}, one can find a solution
or correctly report that none exists in polynomial working space and
deterministic time
\[
 \operatorname{poly}(N,m)(3/2)^m.
\]
\end{corollary}

\begin{proof}
For $m\ge1$, construct the deterministic cover $\calQ_m$ of
Theorem~\ref{thm:streamcover}. For each center, form the affine system
\eqref{eq:generalrows} and test consistency. If it is consistent, return
one particular solution; each of its equations implies the corresponding
product constraint. Conversely, every solution of the product system has
a legal state in some covering cube, so at least one of these affine
systems is consistent whenever a solution exists. If all systems are
inconsistent, report that there is no solution. Cover generation and one
elimination per center give the bounds. For $m=0$, return any assignment.
\end{proof}

This search algorithm needs only consistency tests. Complete enumeration
also pays for every solution of every covering system. Since the product
constraints here are fixed input, the diagonal choice of
Section~\ref{sec:parity} is unavailable. We instead translate the entire
cover. For $m\ge1$, construct $\calQ_m$ and draw a uniform
$t\in\Zthree^m$. Use centers $q+t$, $q\in\calQ_m$, where arithmetic
is coordinatewise modulo three. For a legal state $s$, assign it to
the translated center
\begin{equation}\label{eq:shiftowner}
 \Own_t(s)=\Own(s-t)+t.
\end{equation}
Translation preserves coordinatewise inequality, so this is a covering
center for $s$.

\begin{lemma}[Translation controls repeated visits]\label{lem:shift}
For a fixed affine product system with $L$ solutions, let $I(t)$ be the
sum of the numbers of solutions of all affine systems associated with
the translated centers. Then
\begin{equation}\label{eq:translation}
 \E_t I(t)=|\calQ_m|(2/3)^mL\le2mL.
\end{equation}
\end{lemma}

\begin{proof}
For each fixed center $q$ and each fixed solution $\xi$, let $s(\xi)$
be its legal state. The affine system for $q+t$ contains $\xi$ exactly
when $q_i+t_i\ne s_i(\xi)$ for all $i$. Each condition has probability
$2/3$, independently across coordinates. Summing this probability over
the $|\calQ_m|L$ pairs proves the equality, and~\eqref{eq:globalsize}
proves the inequality. No independence between different centers or
different solutions is required.
\end{proof}

\begin{proof}[Proof of Theorem~\ref{thm:general}]
Construct $\calQ_m$ and choose $t$ as above. Regenerate the original
centers $q$ one at a time, form the affine system~\eqref{eq:generalrows}
for $q+t$, and stream its solutions $\xi$. Compute $s(\xi)$ and output
$\xi$ if and only if $q=\Own(s(\xi)-t)$.

Every affine solution is a solution of the input product system by
\eqref{eq:generalrows}. For every solution $\xi$, its state belongs
to the translated cover, and its owner specifies exactly one affine
system in which it is output. Distinct assignments may have the same
state; they are still distinct solutions of that affine system and are
each output once. Correctness therefore holds for every translation $t$.

Construction, center generation, and elimination cost
$\operatorname{poly}(N,m)(3/2)^m$ in total. Each of the $I(t)$ visits
requires polynomial additional time, so Lemma~\ref{lem:shift} yields
the claimed expected bound. Uniform ternary symbols can be generated
from fair bits by rejection, with constant expected work per symbol.
The cover, translation, current affine system, and streamed solution
require polynomial space. If $m=0$, simply stream all $2^N=L$
assignments. This also satisfies the stated bound.
\end{proof}

The only randomness in the enumeration algorithm is the translation; it
affects running time alone. Feasibility, in contrast, is already deterministic
by Corollary~\ref{cor:decision}.

\section{Limits of complete enumeration}
\label{sec:limits}

There are two different barriers for the present method. A universal legal
affine cover needs at least $(3/2)^n$ pieces. More strongly, some input
fibers contain that many solutions even after the best diagonal choice.

\begin{proposition}\label{prop:coverlower}
The minimum number of affine subspaces wholly contained in $\Sigma^m$
whose union is $\Sigma^m$ equals the binary-subcube covering number.
In particular it is at least $\lceil(3/2)^m\rceil$.
\end{proposition}
\begin{proof}
Each coordinate-pair projection of a nonempty affine subspace is affine,
so it has one, two, or four points. If it is contained in $\Sigma$, it has
at most two. Enlarge each such projection to two legal points; their
product is a binary subcube containing the original subspace. This converts
any legal affine cover to a binary-subcube cover of no greater size.
Conversely, binary subcubes are legal affine spaces. The volume bound is
$3^m\le M2^m$.
\end{proof}

The proposition concerns the entire legal state space, not one input fiber.
To bound complete enumeration of a fiber, put $a=3/2$ and recall $K_B(c)$
from~\eqref{eq:PB}. For every fixed $B$,
\begin{equation}\label{eq:fixedBmean}
 \sum_{c\in\F^n}K_B(c)=\sum_{x\in\F^n}2^{n-|x|}=3^n,
 \qquad \min_cK_B(c)\le a^n.
\end{equation}

\begin{theorem}\label{thm:lower}
For every $n\ge36$ there is a strongly connected zero-diagonal digraph
$B$ with $\min_cK_B(c)\ge(4/27)a^n$. Moreover,
\[
 \left(\frac{1024}{2187}-o(1)\right)a^n
 \le\max_{\substack{B\text{ zero-diagonal}\\ B\text{ strongly connected}}}
       \min_cK_B(c)\le a^n.
\]
\end{theorem}

\begin{lemma}[Joint second moment]\label{lem:moment}
Let $D$ be a random zero-diagonal $m\times m$ matrix whose off-diagonal
entries are independent uniform bits, and independently let $c$ be
uniform in $\F^m$. With $\mu=(3/2)^m$,
\begin{equation}\label{eq:jointmoment}
 \E_{D,c}K_D(c)=\mu,\qquad
 \Var_{D,c}K_D(c)=\mu-(5/4)^m\le\mu.
\end{equation}
\end{lemma}

\begin{proof}
The matrix $A=D+\diag(c+\one)$ has all its entries independent and
uniform. Let $I_x$ indicate that $x\circ(Ax)=x$.
For each active row $i\in\supp x$, feasibility requires the one
nontrivial equation $A[i,:]x=1$. Different rows are independent, so
$\E I_x=2^{-|x|}$, also for $x=0$, when $I_0=1$.

If $x,z$ are distinct and nonzero, they are linearly independent over
$\F$. A row active in both feasibility events then imposes two
independent linear conditions; a row active in just one imposes one.
Independence across rows gives
$\E(I_xI_z)=2^{-|x|-|z|}=\E I_x\E I_z$.
The constant $I_0$ has zero covariance with every variable.
Thus the indicators are pairwise independent, and
\[
 \Var\sum_x I_x
 =\sum_x\bigl(2^{-|x|}-4^{-|x|}\bigr)
 =(3/2)^m-(5/4)^m.
\]
The expectation follows by the same binomial sum.
\end{proof}

\begin{lemma}[A control space avoiding bad fibers]\label{lem:controlspace}
Let $r,m\ge1$ and $0<\varepsilon<1$ satisfy
$(2^r-1)\varepsilon^{-4}(8/9)^m<1$. There are a zero-diagonal matrix
$D\in\F^{m\times m}$ and a linear map $T:\F^r\to\F^m$ such that,
for every $c$, at most one $z\in\F^r$ satisfies
$K_D(c+Tz)<(1-\varepsilon)a^m$.
\end{lemma}
\begin{proof}
Let $F_D=\{c:K_D(c)<(1-\varepsilon)a^m\}$. Chebyshev's inequality and
Lemma~\ref{lem:moment} give
\[
 \E_D|F_D|\le 2^m\varepsilon^{-2}a^{-m}
 =\varepsilon^{-2}(4/3)^m.
\]
Choose $D$ attaining this upper bound, and set $F=F_D$. For a uniform
random linear map $T$ and any fixed nonzero $z$, $Tz$ is uniform in
$\F^m$. Therefore
\[
 \Pr[\exists z\ne0:Tz\in F+F]
 \le(2^r-1)|F|^2/2^m
 \le(2^r-1)\varepsilon^{-4}(8/9)^m<1.
\]
Fix a map avoiding these differences. Two distinct preimages of bad
syndromes would give $T(z+z')\in F+F$, a contradiction. The argument
also covers $F=\varnothing$.
\end{proof}

\begin{proof}[Proof of Theorem~\ref{thm:lower}]
Fix a zero-diagonal controller $G\in\F^{r\times r}$, and put
$k_G=\min_dK_G(d)$. Take $D,T$ from Lemma~\ref{lem:controlspace} and form
\[
 \widetilde B=\begin{pmatrix}D&T\\0&G\end{pmatrix}.
\]
For any syndrome $(c,d)$, its solution count is exactly
\[
 K_{\widetilde B}(c,d)=\sum_{z\in\mathcal P_G(d)}K_D(c+Tz)
 \ge(k_G-1)(1-\varepsilon)a^m.
\]
There are at least $k_G$ terms and at most one is bad. Add one vertex $h$
with arcs in both directions to every old vertex. The result is strongly
connected. Setting $x_h=0$ extends every old solution for every value of
the new syndrome bit. With $n=m+r+1$, the same lower bound thus holds
uniformly over all diagonals of the new graph.

For the finite bound, take a bidirected edge as controller. Its four
syndrome counts are $3,2,2,2$, so $r=2$ and $k_G=2$. With
$\varepsilon=1/2$, the condition is $48\,8^m<9^m$, which holds for all
$m\ge33$. Thus $n\ge36$ and the bound is
$\tfrac12a^{n-3}=(4/27)a^n$.

For the asymptotic constant, let $G$ consist of two disjoint bidirected
triangles. A single triangle has solution counts $5,3,3,4$ according as
the syndrome weight is $0,1,2,3$: an odd nonempty support lies in the zero
coordinates and an even support lies in the one coordinates. Hence
$r=6$ and $k_G=9$. Choose $\varepsilon=1/m$; the lemma's condition holds
for all sufficiently large $m$. The resulting lower bound is
\[
 8(1-1/m)a^{n-7}
 =\left(\frac{1024}{2187}-o(1)\right)a^n.
\]
The upper bound follows from~\eqref{eq:fixedBmean}.
\end{proof}

Thus explicitly visiting every P2 solution has worst-case cost
$\Omega((3/2)^n)$ even if an optimal diagonal is supplied for free;
omitting $x=0$ changes the count by at most one. The witness argument
is existential and does not claim an efficient construction. This is not
a lower bound for evaluating the weighted parity sum: structural shortcuts
or aggregation may avoid complete enumeration. Strong connectivity rules
out only the immediate rejection of a graph that is not strongly connected.

\section{Discussion}
\label{sec:discussion}

The smaller exponential base follows from combining a regenerable cover
with ownership and a budget for every affine solution visit. An improvement
beyond this base would need to avoid the complete-enumeration barrier, for
example by summing weighted contributions in aggregate or by constructing
covers adapted to an input fiber. An even-size affine solution space need
not contribute zero, since both its P3 weights and ownership tests can vary.

For general affine product systems, deterministic feasibility already has
the same exponential bound in the number of constraints. Derandomizing
the output-sensitive enumeration remains open here: fixing part of a
ternary translation does not give the simple rank-based conditional counts
available for the Hamiltonian diagonal choice.

\bibliographystyle{plainurl}
\bibliography{references}
\end{document}